\documentclass[onecolumn]{autart}    % Enable this line and disable the 
\usepackage{mathptmx} % assumes new font selection scheme installed
\usepackage{times} % assumes new font selection scheme installed
\usepackage{amsmath} % assumes amsmath package installed
\usepackage{amssymb}  % assumes amsmath package installed
\usepackage{graphicx}
\usepackage{xcolor,comment}
\usepackage{float}
\usepackage{color}
\usepackage{graphicx}          % Include this line if your 
\newtheorem{theorem}{Theorem}
\newtheorem{proposition}{Proposition}
\newtheorem{lemma}{Lemma}

\newtheorem{remm}{Remark}
\newenvironment{remark}{\begin{remm}\rm }{\hfill \hspace*{1pt} \hfill $\lrcorner$\end{remm}}
\newenvironment{proof}{\noindent {\em Proof.}}{\hfill \hspace*{1pt}\hfill $\square$\\}

\begin{document}

\begin{frontmatter}

\title{A Reference Governor for linear systems with polynomial constraints} % Title, preferably not more 

\thanks[footnoteinfo]{This paper was not presented at any IFAC 
	meeting. Corresponding author L.~Burlion. Tel. +1 848 445 2046. 
}

\author[Paestum]{Laurent Burlion}\ead{laurent.burlion (at) rutgers.edu},    
\author[Paestum]{Rick Schieni},        
\author[Rome]{Ilya Kolmanovsky}

\address[Paestum]{Rutgers University, Piscataway, New Jersey 08854}  % Please supply                                              
\address[Rome]{University of Michigan, Ann Arbor, Michigan 48109}             % full addresses
             % full 

                             % keyword list or with the 
                                          % help of the Automatica 
                                          % keyword wizard

\begin{abstract}                          % Abstract of not more than 200 words.

The paper considers the application  of reference  governors to linear discrete-time  systems with  constraints given by polynomial inequalities.  We propose a novel algorithm  to compute the maximal output admissible  invariant set in the  case of  polynomial  constraints. The reference governor solves a  constrained nonlinear minimization problem at initialization and then uses a bisection algorithm at the subsequent time steps. The effectiveness of the method is demonstrated by two numerical examples.

\end{abstract}

\end{frontmatter}

\section{Introduction}

Reference governors are add-on schemes that, whenever possible, preserve the response of a nominal controller designed by conventional control techniques \cite{GARONE2017306} while ensuring that output constraints are not violated. Conventional reference governor schemes are based on the so called maximal output admissible set (i.e, the set of initial conditions and constant reference signals which guarantee present and future output constraint satisfaction. Finitely determined invariant inner approximations of this set can be easily computed when both the systems dynamics and output constraints are linear.\\ Of particular importance in control theory is the case when the nominal controller is based on feedback linearization (see, e.g, \cite{Isidori1995,Krener2015}) and the nonlinear dynamics are rendered linear by a coordinate transformation and an appropriately defined feedback law. However, when input or output constraints are present, even if they were linear to start with, they are generally transformed into nonlinear constraints on the transformed coordinates. Therefore, in general, feedback linearization cannot be combined with a conventional reference governor \cite{GARONE2017306} which assumes a linear model and linear constraints. A standard reference governor attached to a typical plant is illustrated in Figure \ref{fig_image_RefGov}.
\begin{figure}[!ht]
	\begin{center}
		\includegraphics[width=8cm]{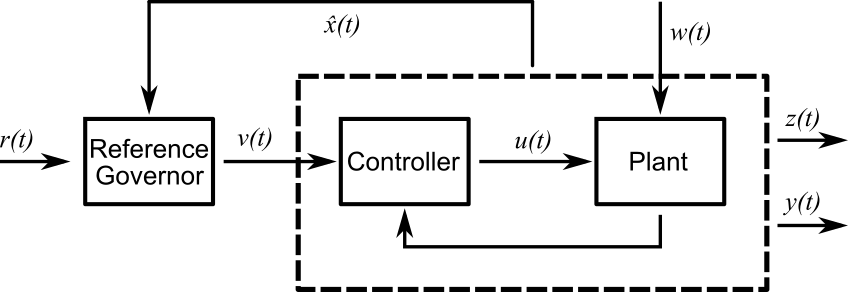}
		\caption{A typical reference governor attached to plant}
		\label{fig_image_RefGov}
	\end{center}
\end{figure}
\begin{comment}
\textcolor{blue}{A few words on the RG synthesis for nonlinear systems \cite{KalabicSO32017},	\cite{KalabicCDC2011}, \cite{KalabicSE3ACC2016}}
\textcolor{red}{\cite{AngeliIJRNC1999} proposes a command governor (CG) design
based on embedding the nonlinear system model into a family of
Linear Time Varying polytopic uncertain models. A similar idea
is used in \cite{FalugiIET2005}, where the nonlinear
system is embedded into a Linear Parameter Varying (LPV) system
that is controlled
}
\textcolor{red}{bisection algorithms were previously used... \cite{KolmanovskyNMPC2012} and \cite{KalabicSO32017}}
\\\\
 \textcolor{blue}{Here, we're interested by the case where the dynamics is polynomial...remark that there is now a vast literature on the control of polynomial systems...and there are also some papers which use these kind of tools even when the dynamics is not polynomial...}
\end{comment}
\\
In this paper, the problem of designing a reference governor for linear systems with polynomial constraints is addressed. The key idea is to embed the linear system into another higher dimensional linear system, the state of which, when correctly initialized, encompasses the state of the original linear system plus its higher order powers. Doing so, the polynomial inequality constraints required to design a reference governor become linear with respect to the extended state's coordinates.\\ The contributions of this paper are as follows. Firstly, we develop a procedure for the computation of the maximal output admissible set (MOAS) of a linear system with polynomial constraints. It is shown that this set is a subset of the aforementioned extended linear system with some linear constraints and that it is finitely determined under suitable conditions. The second contribution of this paper is the design of the reference governor utilizing thereby computed MOAS. This design, unlike conventional reference governors, exploits exponentially decaying reference dynamics so that the system cannot be stuck on a constraint. Then, the reference governor computation requires solving  a constrained nonlinear minimization problem at the initialization and then using a bisection algorithm at the subsequent time instants. Finally, numerical examples are reported to illustrate the proposed approach.
\section{Preliminaries and problem statement}

In the sections that follow the reference governor design will employ the use of the Kronecker product \cite{LOAN200085}. For that reason, this section first introduces the Kronecker product followed by a problem statement.

\subsection{Kronecker products and polynomial systems base vectors}
The Kronecker product of matrices A and B is denoted by $A \otimes B$. This product is non-commutative but associative and has the following useful mixed product property: if $A, B, C$ and $D$ are matrices of such size that one can form the matrix products $AB$ and $CD$, then
\[
(A B \otimes C D) =(A \otimes C)(B \otimes D).  
\]
Given a vector $x\in\mathbb{R}^{n_x}$, another vector $x^{p\otimes}\in\mathbb{R}^{p n_x}$ is defined by:
\begin{eqnarray}
x^{p\otimes}&:=&\underset{i=1..p}{\otimes} x = x\otimes (x \otimes\ldots (x\otimes x)) \\
&=& \begin{bmatrix}
x_1 x^{(p-1)\otimes}&
x_2 x^{(p-1)\otimes}&
\ldots &
x_{n_x}x^{(p-1)\otimes}
\end{bmatrix}^T
\end{eqnarray}
Since the product of two real numbers is commutative, it is not difficult to see that the vector $x^{p\otimes}$ possesses some redundant terms. We use the notation $x^p$ to denote a base vector containing all monomials $x_1^{i_1}\ldots x_{n_x}^{i_{n_x}}$ for which $i_1+\ldots +i_{n_x}=p$. Then, and as first remarked in \cite{ChesiTAC2003}, the dimension of the non-redundant vector $x^p$ is given by
\begin{equation}
\label{eq_sigma}
\sigma(n_x,p)=\frac{(n_x+p-1)!}{p! (n_x-1)!}.
\end{equation}
Following \cite{ValmorbidaTAC2013}, one can compute $x^p$ from $x^{p \otimes}$ using the following relation,
\[
x^p = M_c(n_x,p)x^{p \otimes},
\]
where $M_{c}(n_x,p)\in \mathbb{R}^{\sigma(n_x,p) \times p n_x}$.
Conversely, one can compute $x^{p \otimes}$ from $x^p$ using
\[
x^{p \otimes} = M_e(n_x,p)x^{p},
\]
where $M_{e}(n_x,p)\in \mathbb{R}^{ p n_x \times\sigma(n_x,p)}$. Note that for $i,j\in\mathbb{N}$, $M_c(i,j)$ and $M_e(i,j)$ are computed using an iterative algorithm which is reported in Appendix \ref{eq_Mc_Me}.
\\\\
For example, when $n_x=p=2$, we have $x=[x_1,x_2]^T$, $x^{2\otimes}=[x_1^2,x_1 x_2, x_2^2, x_2 x_1]^T$ and $x^{2}=[x_1^2,x_1 x_2, x_2^2]^T$.

\subsection{Problem statement}

Consider now a linear (pre-stabilized) discrete-time system with the model given by
\begin{equation}\label{class_sys}
x(k+1) = A x(k) + B v(k), 
\end{equation}
where $x\in\mathbb{R}^{n_x}$ is the state, $v\in\mathbb{R}^{n_v}$ is the reference governor output, and $A$ is a Schur matrix. Suppose that the desired reference is $r(k)=0$ and the input into (\ref{class_sys}), $v(k)$, is generated by:
\begin{equation}\label{eq_v}
v(k+1) = \lambda v(k),
\end{equation}
where $\lambda\in]0,1[$. 
\begin{remark}
	Note that $\lambda$ is equal to 1 in the classical reference governors \cite{GilbertTanTAC1991,Gilbert99fastRG}, however, $\lambda<1$ is also used in \cite{kalabiccdc2014} to handle the cases when the reference and/or the constraints are time-varying. Here, choosing $\lambda$ in $]0,1[$ will be useful in the sequel since, otherwise, the polynomial constraints could prevent $v(k)$ from tending to $r=0$ as $k\to\infty$.
\end{remark}
Let $x_v:= [x^T,v^T]^T \in\mathbb{R}^{n_x+n_v}$, be the state and reference vector which evolves according to
\begin{equation}\label{eq_xv}
x_v(k+1) = \Phi x_v(k), 
\end{equation}
where $\Phi=\begin{bmatrix}A & B \\O_{n_v,n_x} & \lambda I_{n_v}\end{bmatrix}$ is a Schur matrix since $\lambda\in]0,1[$.\\
Using the aforementioned notations, system (\ref{eq_xv}) is subject to $n_c$ polynomial constraints, which are expressed as follows
\begin{equation}\label{eq_poly_cons}
\sum_{j=1}^{\sigma(n_x+n_v,p)} c_{i,j} x^{j}_v \leq h_i, \quad i\in\{1,\ldots,n_c\},
\end{equation}
where the row matrices $c_{i,j}$ are in $\mathbb{R}^{1 \times \sigma(n_x+n_v,j)}$ and where, without any loss of generality, $h_i\geq 0$ for all $i\in\{1,\ldots,n_c\}$.
\begin{remark}
	Adding polynomial constraints is thus an operation very similar to adding linear constraints, except that we had first to specify a basis to represent the polynomial constraint as linear constraints.
\end{remark}
In this paper, we propose a reference governor strategy to ensure that the polynomial constraints (\ref{eq_poly_cons}) of a pre-stabilized linear system (\ref{class_sys}) are satisfied for all time while the reference governor output tends to the desired reference $r=0$. In the sequel, we propose a method both to compute the maximal output admissible set (i.e the set of initial states and references which guarantee present and future output constraint satisfaction) and to update the reference governor based on such a set.
\begin{remark}
For simplicity, only the desired $r=0$ reference is considered in this paper. When, $r\neq0$, (\ref{eq_v}) becomes $v(k+1)=\lambda v(k) + (1-\lambda)r$ and the state and reference vector to consider would be $x_{v,r}=[x^T,v^T,r^T]^T$ with $r(k+1)=r(k)=r$. The design proposed in the sequel still applies with the difference that $x_{v,r}$ replaces $x_v$.
\end{remark}
\section{Reference governor design}

The conventional reference governors (RG) are based on the computation of a slightly tightened version of the maximal output admissible set $O_\infty$ (MOAS) \cite{GilbertTanTAC1991,Gilbert99fastRG,GARONE2017306}.  Usually this set, denoted by $\tilde{O}_\infty$, is computed off-line, and the reference $v$ is updated on-line by solving a linear programming problem which is based on this inner approximation. Hence, the computational effort is generally small which is a strength of conventional RG. The objective of this section is to extend these ideas to linear systems subjected to polynomial inequality constraints. 
As such, we propose a new procedure
\begin{itemize}
	\item to compute (off-line) the maximal output admissible set $O_\infty$ for system (\ref{eq_xv}) with constraints (\ref{eq_poly_cons}),
	\item to update the reference governor online based on this set.
\end{itemize}
\begin{remark}
	As $\Phi$ in (\ref{eq_xv}) is Schur, the actual maximum output admissible set can be computed without requiring inner approximation. 
\end{remark}

% \textcolor{red}{the reference governor $v$ is computed on-line using this time a bisection algorithm. the price to pay is the initialization of the $v$ but this can also be done off-line using a grid...}
\subsection{Maximal output admissible set computation}

Let $p$ be given and consider the following state augmentation:
\[
X_v := \begin{bmatrix}x_v & x_v^2 & \ldots & x_v^p\end{bmatrix}^T
\]
Let $j\in\{1,p\}$ and observe that:
\begin{eqnarray}
\label{eq_xj_v}
x^j_v(k+1) &=& M_c(n_x+n_v,p)x_v^{j\otimes}(k+1) \nonumber \\&=&  M_c(n_x+n_v,j)(\otimes_{i=1:j}\Phi)M_e(n_x+n_v,j)x^{j}_v(k)\nonumber\\&:=&\Phi^jx^{j}_v(k)
\end{eqnarray}
i.e., the extended state vector of all monomials can be linearly propagated.
Before stating our main result, we require the following two lemmas.
\begin{lemma}
	if $\Phi$ is a Schur matrix then $\Phi^j$ is a Schur matrix for all $j\in\mathbb{N}^*$. 
\end{lemma}
\begin{proof}
	The proof readily follows from the fact that if $\Phi$ is Schur, then $x_v=0$ is a globally asymptotically stable equilibrium of (\ref{eq_xv}). As a consequence, for all $j\in\mathbb{N}^*$, $x_v^j=0$ is a globally asymptotically stable equilibrium of (\ref{eq_xj_v}), which in turn implies that $\Phi^j$ is a Schur matrix.
\end{proof}
Let
\begin{eqnarray}
\Phi_Z&=&diag(\Phi^j, j\in\{1\ldots p\}) \quad,\quad H_Z=[h_1\ldots h_{n_c}]^T, \\
C_Z &=&(c_{i,j})_{i=1:n_c,j=1:\sigma(n_x+n_v,p)}, 
\end{eqnarray}
and consider the following extended system:
\begin{equation}\label{eq_sysZ}
Z(k+1)= \Phi_Z Z(k)
\end{equation}
where $Z_j\in\mathbb{R}^{\sigma(n_x+n_v,j)}$ and where $Z=[Z_1^T,\ldots,Z_p^T]^T \in\mathbb{R}^{\sum_{j=1 \ldots r}\sigma(n_x+n_v,j)}$ is subject to the constraints,
\begin{equation}\label{cons_Z}
y_Z:= C_z Z \leq H_Z .
\end{equation}
\begin{remark}
	Note that for particular initial conditions i.e $Z_{i}(0)=Z^i_1(0)=x_v^i(0)$ for all $i\in\{2,p\}$, $Z_1$ evolves according to our model (\ref{eq_xv}) and the polynomial constraints (\ref{eq_poly_cons}) exactly match the linear constraints (\ref{cons_Z}). 
\end{remark}
At this stage, we can also add additional constraints to system ($\ref{eq_sysZ}$) since some components of $Z$ must be necessarily positive when we impose $Z_{i}(0)=Z^i_1(0)$ and  $i$ is even. These new inequalities can be obtained using the Kronecker product.
Indeed, for each $e_{i,j}$ basis vector of $\mathbb{R}^{\sigma(n_x+n_v,j)}$ with $j^2\leq p$, we can impose the additional constraint $-e_{i,j}^2\leq 0$ by adding to $C_Z$ the line \begin{eqnarray}\Big[O_{(1,\sum_{k=1:j^2-1}\sigma(n_x+n_v,k))} ~,~ -(M_c(n_x+n_v,2) e_{i,j} \otimes e_{i,j})^T\nonumber \\~,~ O_{(1,\sum_{k=j^2+1:p}\sigma(n_x+n_v,k))}\Big],\nonumber\end{eqnarray}

and to $H_z$ the 'line' $0$.
\begin{lemma}
	Let $\Phi_Z$ be a Schur matrix, $\{y_Z,~ y_Z \leq H_Z \}$ be a compact set, and let the pair $(\Phi_Z,C_Z)$ be observable then $O_{\infty,Z}$, the maximal output admissible set of (\ref{eq_sysZ})-(\ref{cons_Z}), is finitely determined and is forward invariant.
\end{lemma}
\begin{proof}
	The proof readily follows from the classical arguments of \cite{GilbertTanTAC1991}. 
\end{proof}
Let us now focus our interest on the subset of $O_{\infty,Z}$ such that $Z_i= Z^i_1$ for all $i\in\{2,p\}$. We've got the following result:
\begin{theorem}\label{thm1}
	Let $\Phi_Z$ be a Schur matrix, $\{y_Z,~ y_Z \leq H_Z \}$ be a compact set, and let the pair $(\Phi_Z,C_Z)$ be observable. Then $O_{\infty,X}$, the maximal output admissible set of (\ref{eq_xv})-(\ref{eq_poly_cons}), is finitely determined and is forward invariant.
\end{theorem}
\begin{proof} it is easy to see that $O_{\infty,X}$ is the following subset
	\[O_{\infty,X}=\{Z\in O_{\infty,Z} ~s.t~Z_i=Z_1^i \}.\]
As such, it is finitely determined under our assumptions, as follows from \cite{GilbertTanTAC1991}.
Let $Z(0)\in O_{\infty,X}$, then for all $k\in\mathbb{N}$, $Z(k)\in O_{\infty,Z}$ since $O_{\infty,X} \subset O_{\infty,Z}$, which is forward invariant. However, $Z_i(k)=Z_1(k)^i$ since $Z(k+1)=\Phi_Z Z(k)$. So, for all $k\in\mathbb{N}$, $Z_i(k) \in O_{\infty,Z}$.
\end{proof}
Theorem \ref{thm1} requires that the pair $(\Phi_Z,C_Z)$ be observable which may not hold when the dimension of $\Phi_Z$ is very large and the  number of constraints $n_c$ is small. The following result provides a relaxed but more practical way to satisfy these requirements. 
\begin{proposition}
Let $H^1_Z=H_Z(1:n_x+n_v,1)$. Suppose there exist	$C^1_Z \in \mathbb{R}^{n_c \times (n_x+n_v)}$ such that
$\{Z=[Z_1,Z_1^2,\ldots,Z_1^p],~C_Z Z \leq H_Z\} \subset \{Z_1,~C^1_Z Z_1 \leq H^1_Z\}$, a convex compact set, and that $(\Phi,C^1_Z)$ is observable. Then the maximal output admissible set of (\ref{eq_xv})-(\ref{eq_poly_cons}) is finitely determined and is forward invariant.
\end{proposition}
\begin{proof}
	Simply observe that in this case the set $O_{\infty,Z_1}$ is finitely determined and compact. Then, the maximum and minimum values of $Z_1(i)$ can be determined solving a linear programming problem. Then, given these bounds, one can compute some bounds for each components of $Z_1^i$ for all $i\in\{1,\ldots,p\}$. One can add these bounds as new rows in $C_Z$ and $H_Z$. These new linear constraints are denoted by $C_Z^a$ and $H_Z^a$. With each higher order component being bounded, it is easy to see that $(\Phi_Z,C^a_Z)$ is observable if $(\Phi,C^1_Z)$ is observable. Finally, the proof is completed by applying Theorem \ref{thm1} where $C^a_Z$ replaces $C_Z$.
\end{proof}

\subsection{Reference governor update}

Let $O_{\infty,X}$ be the maximal output admissible set defined in the previous section. 
Given an initial state $x(0)$, $v(0)$ is computed as a solution to the following optimization problem
\begin{eqnarray}
&\underset{v}{\min}& v^T v \\
&s.t& (x(0),v)\in O_{\infty,X}
\end{eqnarray}
\begin{remark}
Before solving this nonlinear optimization problem, we recommend eliminating the redundant inequalities from the representation of $O_{\infty,X}$ for the particular choice of $x=x(0)$.
\end{remark}
Then $v$ is computed at each next time instant using a bisection algorithm %(see Appendix \ref{bisection} for more details). 
Note that $v(k)$ tends to $0$ because $\lambda\in]0,1[$. The preview time of the bisection algorithm does not need to be 'heuristically' chosen since it is directly applied to the set of inequalities that defines the MOAS. To reduce the computational burden and thus facilitate the practical application of the proposed method, we suggest to calculate 'off-line' both the MOAS and the initial value $v(0)$ on a grid of initial conditions. Doing so, only a simple bisection algorithm needs to be run 'on-line'.

\begin{comment}
\section{Other practical 'hints' (to be deleted)}
\subsection{Side problem: fixed preview time of a bisection algorithm}
It may be possible to shrink $O_\infty$ so that one can use a given preview time...
\subsection{Pre and Post processing in practice}
For $v$ constant, observe that the dynamics of $X_v^r$ is completely independent of the one of $X_v^{j}$ where $j\neq r$... so multiple time scales can be used...
\begin{itemize}
	\item one can search for a shorter horizon: in this case one can compute a set using less inequalities, use a tolerance to reduce it and then prove that it is included in $O_\infty$.
	\item for the initialization of the bisection, one can compute $v$ on a smaller horizon and then check that it satisfies all the inequalities
\end{itemize}
\end{comment}

\section{Numerical examples}

\subsection{Stall prevention of a civil aircraft}

We consider the following aircraft longitudinal dynamics model based on \cite{NicotraCM2018} with $\cos(\alpha)$ approximated by 1:
\[
\ddot{\alpha} = -\frac{d_1}{J}L(\alpha)+ \frac{d_2}{J} u,
\]
where
\[
L(\alpha)=l_0+l_1 \alpha -l_3 \alpha^3,
\]
and $d_1=4m$, $d_2=42m$, $J=4.5\times10^5 Nm^2$, $l_0=2.5\times10^{5}N$, $l_1=8.6\times10^{6}N/rad$ and $l_3=4.35\times10^7 N/rad^3$.
The angle of attack $\alpha$ is constrained by the stall limit as $-0.2\times\frac{\pi}{180}\leq \alpha\leq 14.7\times\frac{\pi}{180}$rad. The control input is the elevator force $u$ and must satisfy $|u|\leq 4. 10^5$ N.\\
\\
Applying a dynamic inversion, $u=-k_p(\alpha+v)-k_d \dot{\alpha}+\frac{d_1}{d_2}L(\alpha)$ where $k_p =5.2\times10^7$, $k_d = 7.6\times10^6$, and discretizing the system with sampling period $T_s=0.01s$, we obtain the (pre-stabilized) second order model (\ref{class_sys}) with
\[
A =\begin{bmatrix}
0.9814  &  0.0072\\
-3.3347 &  0.4940 
\end{bmatrix} , B =\begin{bmatrix}
0.0186 \\ 3.3347
\end{bmatrix}, 
\]
This system is linear but the input inequality constraints are polynomial of order 3.
Considering the extended state $x_v$ and $\nu$ defined by (\ref{eq_v}) with $\lambda =.98$, 
we compute the MOAS $O_{\infty,Z_1}$. Note that we only consider the linear constraint on $\alpha$ since the system is observable in this case. The MOAS is finitely determined in 77 iterations and is defined by 107 non-redundant linear inequalities. As this set is compact, we can directly compute some bounds on all the components of the extended state and thus deduce some bounds on its powers. Following the idea of Proposition 1, this allows us to directly check that MOAS $O_{\infty,Z}$ can in turn be finitely determined when we extend the state and add all these constraints. The $O_{\infty,Z}$ is determined in 31 iterations and is defined by 298 non-redundant linear inequalities.
Figure \ref{fig_image_ex1} illustrates the constrained outputs responses obtained using (in blue) or not using (in magenta) the proposed reference governor when $\alpha(0)=14\; deg$ and $\dot\alpha(0)=0\; rad/s$. In the absence of a reference governor, the control input limits are violated. However, all the output constraints are satisfied with the implementation of the proposed reference governor strategy.
The projection of the MOAS set on the $(\alpha,\dot{\alpha})$ plane gives the set of initial conditions from which the state can be stabilized while respecting both the linear and polynomial constraints. Figure 3 shows this set and used a grid to calculate it. The constrained nonlinear minimization problem was solved at each point of the grid to determine whether or not the point belongs to this set.
\begin{figure}[!ht]
	\begin{center}
		\includegraphics[width=8.5cm]{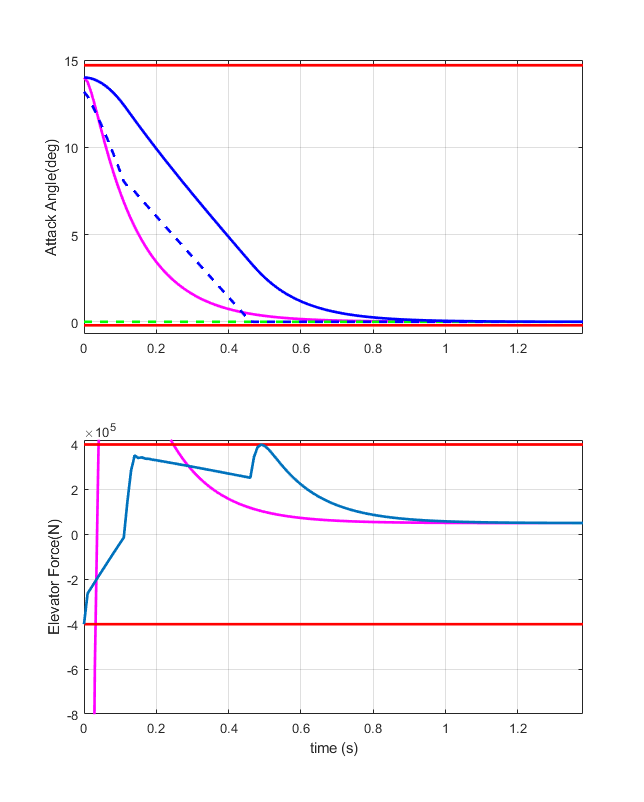}
		\caption{Constrained outputs. Red: upper and lower limitations. Magenta: when one applies the dynamic inversion without any reference governor, that is to say when $v=r=0$. Blue: when one uses the proposed reference governor. Dashed blue: evolution of the reference $v(t)$ of the reference governor. Dashed green: desired angle of attack.}
		\label{fig_image_ex1}
	\end{center}
\end{figure}
\begin{figure}[!ht]
	\begin{center}
		\includegraphics[width=7.5cm]{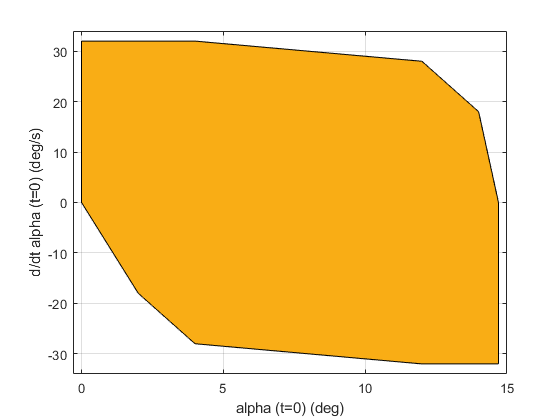}
		\caption{Projection of the MOAS set}
		\label{fig_image_ex12}
	\end{center}
\end{figure}

\subsection{Obstacle avoidance in the 2D plane}

The second numerical example shows the method's applicability to a controlled system subject to a non-convex polynomial constraint. 
We consider $\ddot{x}=u$ where $x,u\in\mathbb{R}^2$. This system is pre-stabilized using $u=-x-2\dot{x}+v$ where $v$ is the reference governor input. Then, we discretize it with sampling period $T_s =0.5$s.
The linear constraints are $|x|\leq 20$ and $|\dot{x}|\leq 5$. The nonlinear constraint $(x-x_{obs})^T (x-x_{obs})\leq r_{obs}^2$ represents a circular obstacle located at $x_{obs}=[10;0]$ with a radius $r_{obs}=2$. 
Considering the extended state $x_v$ and $\nu$ defined by (\ref{eq_v}) with $\lambda =.98$, 
we first compute the MOAS $O_{\infty,Z_1}$. This set is finitely determined in 154 iterations and is defined by 90 non-redundant linear inequalities. Then, $O_{\infty,Z}$ is determined when we extend the state and add the nonlinear constraint. Here, $O_{\infty,Z}$ is determined in 188 iterations and is defined by 217 non-redundant linear inequalities.
Figure 4 illustrates the obstacle avoidance realized with the proposed reference governor strategy when $x(0)=[20;1]$ and $\dot{x}(0)=[0;0]$. Figure 5 shows the evolution of the reference input vector used in this case.
\begin{figure}[!ht]
	\begin{center}
		\includegraphics[width=8.5cm]{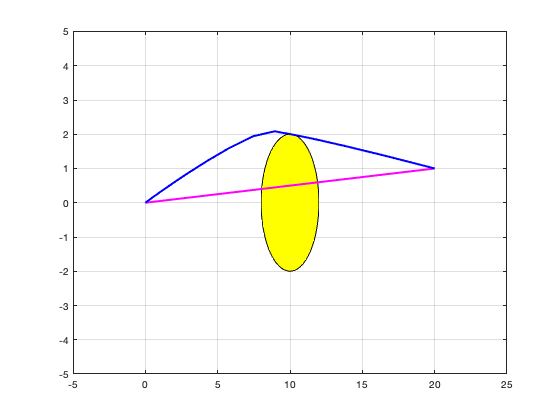}
		\caption{Illustration of the obstacle avoidance. Yellow: obstacle. Magenta: without any reference governor, that is to say when $v=r=0$. Blue: when one uses the proposed reference governor.}
		\label{fig_image_ex21}
	\end{center}
\end{figure}
\begin{figure}[!ht]
	\begin{center}
		\includegraphics[width=7.5cm]{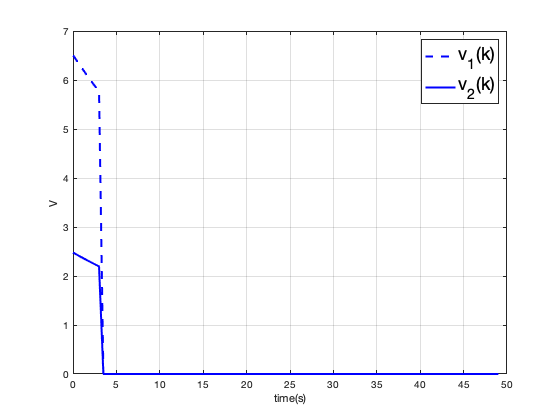}
		\caption{Time evolution of the reference governor}
		\label{fig_image_ex22}
	\end{center}
\end{figure}

\section{Concluding remarks}

The developments in this paper are based on the observation that the propagation of some polynomial constraints through a LTI system can be accomplished by propagating some linear constraints  through a higher dimensional LTI system. This permits extending the design of conventional reference governors to the class of LTI systems with polynomial constraints. Numerical results were reported to demonstrate the simplicity and practicality of the proposed method.

\bibliographystyle{unsrt}
\bibliography{biblioRGpolynomial}

\begin{thebibliography}{10}

\bibitem{GARONE2017306}
E.~Garone, S.~Di Cairano, and I.~Kolmanovsky.
\newblock Reference and command governors for systems with constraints: A
  survey on theory and applications.
\newblock {\em Automatica}, 75(Supplement C):306 -- 328, 2017.

\bibitem{Isidori1995}
Alberto Isidori.
\newblock {\em Nonlinear Control Systems}.
\newblock Springer, London, 1995.

\bibitem{Krener2015}
A.J. Krener.
\newblock Feedback linearization of nonlinear systems.
\newblock In J.~Baillieul and T.~Samad, editors, {\em Encyclopedia of Systems
  and Control}, pages 428--437. Springer London, 2015.

\bibitem{LOAN200085}
Charles~F.Van Loan.
\newblock The ubiquitous kronecker product.
\newblock {\em Journal of Computational and Applied Mathematics},
  123(1):85--100, 2000.
\newblock Numerical Analysis 2000. Vol. III: Linear Algebra.

\bibitem{ChesiTAC2003}
G.~Chesi, A.~Garulli, A.~Tesi, and A.~Vicino.
\newblock Solving quadratic distance problems: an lmi-based approach.
\newblock {\em IEEE Transactions on Automatic Control}, 48(2):200--212, 2003.

\bibitem{ValmorbidaTAC2013}
G.~Valmorbida, S.~Tarbouriech, and G.~Garcia.
\newblock Design of polynomial control laws for polynomial systems subject to
  actuator saturation.
\newblock {\em IEEE Transactions on Automatic Control}, 58(7):1758--1770, 2013.

\bibitem{GilbertTanTAC1991}
E.G. Gilbert and K.T. Tan.
\newblock Linear systems with state and control constraints: the theory and
  application of maximal output admissible sets.
\newblock {\em IEEE Transactions on Automatic Control}, 36(9):1008--1020, Sep
  1991.

\bibitem{Gilbert99fastRG}
E.G. Gilbert and I.V. Kolmanovsky.
\newblock Fast reference governors for systems with state and control
  constraints and disturbance inputs.
\newblock {\em International Journal of Robust and Nonlinear Control},
  9(15):1117--1141, 1999.

\bibitem{kalabiccdc2014}
U.~Kalabi{\'c} and I.~Kolmanovsky.
\newblock Reference and command governors for systems with slowly time-varying
  references and time-dependent constraints.
\newblock In {\em 53rd IEEE Conference on Decision and Control}, pages
  6701--6706, Dec 2014.

\bibitem{NicotraCM2018}
M.M. Nicotra and E.~Garone.
\newblock The explicit reference governor: A general framework for the
  closed-form control of constrained nonlinear systems.
\newblock {\em IEEE Control Systems}, 38(4):89--107, 2018.

\end{thebibliography}

\appendix
\subsection{Computation of $M_c(n,p)$ and $M_e(n,p)$}\label{eq_Mc_Me}

For the reader's convenience, here we review the iterative algorithm which was presented in (\cite{ValmorbidaTAC2013}, Appendix B).\\
First, initialize $\forall p, M_t(1,p)=1$ and $\forall n, M_t(n,1)=I_n$.\\
Then, compute iteratively:
\begin{equation}
M_t(n,p)=\begin{bmatrix}M_t(n,p-1)& O_{n\sigma(n,p-2)\times \sigma(n-1,p)} \\
M_z(n,p) & M_{t0}(n,p)\end{bmatrix}
\end{equation}
with
\begin{eqnarray}
M_z(n,p)&=&\begin{bmatrix}O_{n\sigma(n-1,p-1)\times\sigma(n,p-2)}& I_{\sigma(n-1,p-1)}\otimes e_1
\end{bmatrix} \nonumber \\
M_{t0}(n,p)&=& \left(I_{\sigma(n-1,p-1)} \otimes \begin{bmatrix}O_{1,n-1}\\I_{n-1}\end{bmatrix}\right)M_t(n-1,p)\label{eqMt0}
\end{eqnarray}
Noting $L_t(n,p)=M_t(n,p)^{\dagger}$, one has:
\begin{eqnarray}
M_c(n,p)&=& \Pi_{i=p}^2  (L_t(n,i)\otimes I_{n^{p-i}}) \\&=& L_t(n,p) (L_t(n,p-1)\otimes I_n)\ldots (L_t(n,2)\otimes I_{n^{p-2}})\nonumber\\&&
\end{eqnarray}
\begin{eqnarray}
M_e(n,p)&=& \Pi_{i=2}^p  (M_t(n,i)\otimes I_{n^{p-i}}) \\&=& (M_t(n,2)\otimes I_{n^{p-2}})\ldots(M_t(n,p-1)\otimes I_n) M_t(n,p) \nonumber\\&&
\end{eqnarray}
\begin{remark}
	The expression of $M_{t0}$ retained in \cite{ValmorbidaTAC2013} is slightly more complicated. Instead, we found easier to use (\ref{eqMt0}) which is obtained using the Kronecker product mixed property.
\end{remark}

\end{document}